\documentclass[onecolumn,11pt,letter]{article}
\usepackage{mathptmx}
\usepackage{epsfig}
\usepackage{epstopdf}
\usepackage{graphicx}
\usepackage{amsfonts}
\usepackage{amsmath}
\usepackage{amsthm}
\usepackage{proof}
\usepackage{latexsym}
\usepackage{amssymb}
\usepackage{color}
\usepackage{comment}
\usepackage[ruled]{algorithm2e}
\usepackage[margin=1in]{geometry}
\usepackage{hyperref}
\usepackage{color}
\usepackage{times}
\usepackage{comment}
\usepackage{float}
\usepackage{tweaklist}

\newtheorem{theorem}{Theorem}[section]
\newtheorem{lemma}[theorem]{Lemma}

\newtheorem{definition}[theorem]{Definition}

\DeclareMathAlphabet{\mathcal}{OMS}{cmsy}{m}{n}
\usepackage{tweaklist}
\renewcommand{\enumhook}{\setlength{\itemsep}{-0.5ex}}
\renewcommand{\itemhook}{\setlength{\itemsep}{-0.5ex}}

{\makeatletter
 \gdef\xxxmark{%
   \expandafter\ifx\csname @mpargs\endcsname\relax 
     \expandafter\ifx\csname @captype\endcsname\relax 
       \marginpar{xxx}
     \else
       xxx 
     \fi
   \else
     xxx 
   \fi}
 \gdef\xxx{\@ifnextchar[\xxx@lab\xxx@nolab}
 \long\gdef\xxx@lab[#1]#2{{\bf [\xxxmark #2 ---{\sc #1}]}}
 \long\gdef\xxx@nolab#1{{\bf [\xxxmark #1]}}
 \long\gdef\xxx@lab[#1]#2{}\long\gdef\xxx@nolab#1{}%
}

\newcommand*\samethanks[1][\value{footnote}]{\footnotemark[#1]}

{

\pagestyle{plain}

\begin{document}

\title{A Game-Theoretic Model Motivated by the DARPA Network Challenge\thanks{Dept.\ of Computer Science , University of Maryland, USA.
Email: \tt{\{rchitnis, hajiagha, jkatz, koyelm\}@cs.umd.edu.}}}

\author{Rajesh Chitnis\thanks{Supported in part by NSF CAREER award 1053605, NSF grant CCF-1161626, ONR YIP award
N000141110662, DARPA/AFOSR grant FA9550-12-1-0423, a University of Maryland Research and Scholarship Award (RASA).} \and
MohammadTaghi Hajiaghayi\samethanks[2] \and Jonathan Katz \and Koyel Mukherjee}

\date{January 29, 2013}


\maketitle

\begin{abstract}
In this paper we propose a game-theoretic model to analyze events similar to the 2009 \emph{DARPA Network Challenge}, which
was organized by the Defense Advanced Research Projects Agency (DARPA) for exploring the roles that the Internet and social
networks play in incentivizing wide-area collaborations. The challenge was to form a group that would be the first to find the
locations of ten moored weather balloons across the United States. We consider a model in which $N$ people (who can form
groups) are located in some topology with a fixed coverage volume around each person's geographical location. We consider
various topologies where the players can be located such as the Euclidean $d$-dimension space and the vertices of a graph. A
balloon is placed in the space and a group wins if it is the first one to report the location of the balloon. A larger team
has a higher probability of finding the balloon, but we assume that the prize money is divided equally among the team members.
Hence there is a competing tension to keep teams as small as possible.

\emph{Risk aversion} is the reluctance of a person to accept a bargain with an uncertain payoff rather than another bargain
with a more certain, but possibly lower, expected payoff. In our model we consider the \emph{isoelastic} utility function
derived from the Arrow-Pratt measure of relative risk aversion. The main aim is to analyze the structures of the groups in
Nash equilibria for our model. For the $d$-dimensional Euclidean space ($d\geq 1$) and the class of bounded degree regular
graphs we show that in any Nash Equilibrium the \emph{richest} group (having maximum expected utility per person) covers a
constant fraction of the total volume. The objective of events like the DARPA Network Challenge is to mobilize a large number
of people quickly so that they can cover a big fraction of the total area. Our results suggest that this objective can be met
under certain conditions.
\end{abstract}


\section{Introduction}
\label{sec:intro}

With the advent of communication technologies, and the Web in particular, we can now harness the collective abilities of large
groups of people to accomplish tasks with unprecedented speed, accuracy, and scale. In the popular culture and the business
literature, this process has come to be known as crowdsourcing~\cite{howe-crowdsourcing}. Crowdsourcing has been used in
various tasks such as labeling of images~\cite{luis-von-ahn}, predicting protein structures~\cite{cooper-protein}, and posting
and solving Human Intelligence Tasks in Amazon's \emph{Mechanical Turk}~\cite{pontin-amazon}. An important class of
crowdsourcing problems demand a large recruitment along with an extremely fast execution. Examples of such \emph{time-critical
social mobilization} tasks include search-and-rescue operations in the times of disasters, evacuation in the event of
terrorist attacks, and distribution of medicines during epidemics. For example, in the aftermath of Hurricane Katrina, amateur
radio volunteers played an important role by coordinating dispatch of emergency services to isolated
areas~\cite{krakow-katrina}. ``Collaboratition"~\cite{wiki-collaboratition} is a newly coined term to describe a type of
crowdsourcing used for those problems which require a \textbf{collaborative} effort to be successful, but use
\textbf{competition} as a motivator for the participation or the performance.\\

\textbf{The DARPA Network Challenge:} A good example of collaboratition is  The 2009 \emph{DARPA Network
Challenge}~\cite{darpa-website}, an event organized by the Defense Advanced Research Projects Agency (DARPA) for exploring the
roles that the Internet and social networks play in incentivizing wide-area collaborations. Collaboration of efforts was
required to complete the challenge quickly and in addition to the competitive motivation of the contest as a whole, the
winning team from MIT established what they called a ``collaborapetitive" environment to generate participations in their team
and found all the ten balloons in less than seven hours. Their strategy in Pickard et al.~\cite{pickard2010time}. Their main
focus is on the mechanics of the group formation process in the DARPA Network Challenge, whereas in this paper we try to
analyze the structures of the groups which form in Nash Equilibria.\\

\textbf{Related Work:} Douceur and Moscibroda~\cite{douceur2007lottery} addressed a problem close to the spirit of the DARPA
Network Challenge. They address the problem of motivating people to install and run a distributed service, like peer-to-peer
systems, in which the decision and the effort to install a service falls to the individuals rather than to a central planner.
Their paper appeared in 2007; two years before the DARPA Network Challenge took place. Their focus is on incentivizing the
growth of a single group whereas in this paper we take a bird's-eye view and try to analyze the structures of the groups in
Nash equilibria.


In this paper we focus on analyzing the structures of the groups in Nash equilibria for our model. Some recent results analyze
the structures of Nash Equilibria. Some upper and lower bounds are given on the diameter of the Equilibrium graphs in
\emph{Basic Network Creation Games}~\cite{alon-spaa-10}. It was also shown that the equilibrium graphs have polylogarithmic
diameter in \emph{Cooperative Network Creation Games}~\cite{demaine-stacs-09}. A well-studied parameter related to Nash
equilibria is the \emph{price of anarchy}~\cite{poa-1,poa-2,poa-3}, which is the worst possible ratio of the total cost found
by independent selfish behavior and the optimal total cost possible by a centralized, social welfare maximizing solution.
However as observed in ~\cite{alon-spaa-10,demaine-stacs-09}, bounds on the structures of Nash equilibria lead to approximate
bounds on the price of anarchy as well but not necessarily the other way around. Therefore trying to analyze the structures of
the groups in Nash equilibria is more general than trying to bound the price of anarchy.

Myerson~\cite{myerson1977graphs} used graph-theoretic ideas to model and analyze games with \emph{partial cooperation
structures}. The DARPA Network Challenge is also similar as any groups can possibly form but the geographical locations of the
people causes certain group structures to become infeasible. There is an entire body of literature in Economics which is
closely related to the model we consider in this paper. There have also been studies on how the rules of coalition formation
affect the stability of environmental agreements between countries~\cite{finus}. Their rules for leaving or entering the
coalitions are very similar to the ones we consider in this paper for formation or splitting of the groups.
%
%
Risk aversion is a natural assumption to make while modeling the behavior of humans. There is a recent
paper~\cite{bhalgat-risk-averse} which gives efficient algorithms for computing truthful mechanisms for risk-averse sellers.
Another paper~\cite{bitcoin} considers scenarios in which the goal is to ensure that information propagates through a large
network of nodes. They assume a model where all nodes have the required information to compete which removes the incentive to
propagate information. In this paper, we consider the natural assumption of risk aversion which gives a concave utility
function (derived from the Arrow-Pratt measure of relative risk aversion). This motivates formation of groups which is
consistent with what was observed in the DARPA Challenge. Another paper~\cite{stoc} considers the problem of acquiring
information in a strategic networked environment. They show that in the DARPA Network Challenge, the idea to offer split
contracts instead of fixed-payment contracts is robust against the selfishness displayed by the participating agents.\\

\textbf{Organization of the paper:} In Section~\ref{sec:model} we describe our model in detail. Then we consider various
topologies where our model can be implemented: the one-dimensional (line) space (in Section~\ref{sec:1d}), and more generally
the $d$-dimensional Euclidean space (in Section~\ref{sec:euclidean-d}). For both these topologies, in any Nash Equilibrium we
show that there always exists a group covering a constant fraction of the total volume. In Section~\ref{sec:graph-case} we
consider the discrete version of the our model, where the players form the vertices of an undirected graph. For the class of
bounded-degree regular graphs, we prove that in any Nash Equilibrium there always exists a group covering a constant fraction
of the total number of vertices. In contrast, under an assumption that defecting to an empty group is prohibited, we show for
every constant $0<c<1$ there exist graphs which have a Nash Equilibrium where all groups occupy strictly less than a
$c$-fraction of the total number of vertices.

\section{Our Model}
\label{sec:model}

We assume there is a set of $N$ players, each covering a region of space within the total volume~$A$. In particular, in the
Euclidean space, we assume each player covers a ball of radius one centered at his location; in the discrete case we view the
players as occupying the vertices of a graph and assume each player covers himself and his neighbors.

Players are allowed to organize themselves into a collection of disjoint \emph{groups} partitioning the set of players. In
this work, we do not consider the precise dynamics of group formation, but instead we focus on analyzing the structures of the
groups in Nash equilibria. Once the groups are formed, we envision the balloon being placed in the space.
We say the balloon \emph{falls within} a group~$S$  if the location of the balloon is in the coverage of $S$; a group $S$
\emph{wins} if it is the first one to report the location of the balloon. To model this we assume the probability that the
balloon \emph{falls within} a group~$S$ is $A_S/A$, where $A_S$ is the total volume covered by the players in~$S$ and $A$ is
the total volume. The prize money~$M$ is given to the group that wins, and the money received by a group is split equally
among all members of that group. We note the balloon can be placed anywhere in the space, and we do not know where it will be
placed. Hence the probability of any of the groups (which might form) finding the balloon first is the same and we do not
consider this common factor hereafter.

\emph{Risk aversion}~\cite{holt2002risk,lambert1985risk,tobin1958liquidity,wiki-risk-aversion} is a concept in psychology,
economics, and finance, based on the behavior of humans (especially consumers and investors) whilst exposed to uncertainty.
Risk aversion is the reluctance of a person to accept a bargain with an uncertain payoff rather than another bargain with a
more certain, but possibly lower, expected payoff~\cite{wiki-risk-aversion}. Risk aversion is a natural assumption when we
consider money and people: most of us would accept a guaranteed payment of say X dollars  than a 50\% chance of receiving 2X
and a 50\% chance of getting nothing, especially if X is large (the DARPA Challenge had a prize money of \$40,000).
\emph{Constant relative risk aversion} means that the ratio of the increase in the utility to the increase in the risk taken
is constant. Assuming that the Arrow-Pratt measure of relative risk aversion is constant, the \emph{isoelastic} utility
function for money $x$ is given by $u(x) = \frac{x^{1-r}}{1-r}$ where $0\leq r<1$ is the risk aversion
factor~\cite{holt2002risk}. For $r=1$ we take the utility to be the natural logarithm. Here $r = 0$ means there is no risk
aversion. For simplicity we scale up everything by a factor of $1-r$ to get a concave utility function given by $u(x) =
x^{1-r}$ where $0<r<1$. The \emph{expected} utility for a player who is a member of a group~$S$ is given by $p(S) \cdot
u(\frac{M}{|S|})$, where $p(S) = \frac{A_S}{A}$ is the probability that the balloon fall within $S$. Consider two players who
have disjoint area of coverage. If they are on their own, then their expected utility is $u=(\frac{|M|}{1})^{1-r}\cdot
\frac{a}{A}$ where$a$ is the area they can cover. If they join together to form a group then their expected utility is $u' =
(\frac{|M|}{2})^{1-r}\cdot \frac{2a}{A} = 2^{r}\cdot u > u$ since $1>r>0$. Therefore two people whose coverage areas are
disjoint will always join together, not matter what the risk aversion factor $r$ is. The intuition is that the value of $r$
affects how much overlapping coverage areas is allowed for it to be beneficial for people to join together. The smaller the
value of $r$ the lesser the overlap must be between the coverage areas of the players for it to make sense for them to merge.

We assume that the balloon is placed in a location covered by at least one player.
Given a partition $S_1, \ldots, S_\ell$ of all the players into groups, we now ask whether it forms an equilibrium. More
formally, we allow two types of actions:
\begin{enumerate}
\item Two groups $S_i$ and $S_j$ can decide to \emph{merge}. We say this operation is \emph{incentivized} only if each
    player in $S_i$ and $S_j$ would increase their expected utility by merging.
\item A member $x$ of group $S_i$ may \emph{defect} to join a different group $S_j$. We say this operation is
    \emph{incentivized} only if both $x$'s expected utility and the expected utility of each player in $S_j$ increase
    after the defect.
\end{enumerate}

We say a given partition is a \emph{Nash Equilibrium} if no merge or defect operation is incentivized, i.e., no player can do
better by unilaterally changing his group. We could consider a generalized \emph{defect} operation where a subset of a group
$S_i$ may leave to join a group $S_j$. However for the sake of clarity (while still capturing the essence of the model), we
consider the \emph{defect} operation where at a given time only a single person can leave his current group to join a new
group.



\section{Lower Bounds on the Total Prize Money}
\label{sec:results-and-techniques}

In this paper we consider the \emph{social welfare} from the viewpoint of the agency which hosts the event described by our
model.We now show that the hosting agency needs to offer prize money proportional to the desired size of a largest group or to
the desired fraction of the total volume covered if each person must receive a minimum threshold expected utility.

\begin{theorem}$[\star]$\footnote{The proofs of the results labeled with $\star$ have been deferred to the Appendix}
\label{thm:increasing-area} If there exists a Nash Equilibrium in which at least one group $S$ covers a $\lambda$-fraction of
the total volume, and each player in $S$ covers volume $V$ and has an expected utility of at least~$c$, then $M \geq \lambda A
c^{\frac{1}{1-r}}$.
\end{theorem}

\begin{theorem}$[\star]$
\label{thm:increasing-size} If there exists a Nash Equilibrium in which there is at least one group $S$ of size ~$k$, and each
player in $S$ covers volume $V$ and has an expected utility of at least~$c$, then $M \geq
k\Big(\frac{cA}{N}\Big)^{\frac{1}{1-r}}$.
\end{theorem}

In the rest of the paper the total prize money $M$ is not important as we just compare the expected utilities of various group
structures to see which ones form Nash equilibria, and hence $M$ cancels out. However the above two theorems imply the hosting
agency must spend money $M$ which depends on $c$ and therefore $M$ cannot be arbitrarily small.

\section{The One-Dimensional (Line) Case}
\label{sec:1d}

In this section the players are located along a line. We show for any Nash Equilibrium there is at least one group covering a
constant fraction (depending on the risk aversion factor $r$) of the total length
We assume each person has a coverage length of one on both sides. Recall for each person $x$ in a group $S$ the expected
utility is $E[u(x)] = \Big(\frac{M}{|S|}\Big)^{1-r}\cdot \frac{A_S}{A}$ where $M$ is the total money, $A_S$ is the length
covered by group $S$ and $A$ is the total length. We contract the points not covered by any player. Therefore every point in
the total length has at least one person whose coverage length contains it.

\begin{lemma}
\label{lemma:no-one-has-large-exclusive-area} For the line case, let $S$ be a richest group in a Nash Equilibrium. Then there
is no player $i \notin S$ who can add a length of at least $2(1-r)$ to the length $A_S$ currently covered by $S$.
\end{lemma}
\begin{proof}
Suppose there is a player $i \notin S$, who can add a length of at least $2(1-r)$ to the length covered by $S$. However, since
it is a Nash Equilibrium, either the new expected utility of $S$ on adding this player is less than or equal to the current
expected utility of $S$ (hence $S$ would have no incentive in adding the player $i$) or the player $i$ would not have any
incentive to move to $S$, as the projected new expected utility of $i$ is less than or equal to his current expected utility.
Since $S$ is a richest group, both these conditions combine to give:
\begin{equation}
\label{eq:1d}
 \Big(\frac{M}{|S|+1}\Big)^{1-r}\cdot \frac{A_S + 2(1-r)}{A} \leq \Big(\frac{M}{|S|}\Big)^{1-r}\cdot \frac{A_S}{A}
\end{equation}
As each player has a coverage length of two we have $2|S| \geq A_S$ (equality holds only if the coverage lengths of the
members of $S$ are pairwise disjoint). The function $f(x)=\frac{x}{x+1}$ is increasing on $(0,\infty)$ and hence
$\frac{|S|}{|S|+1} \geq \frac{\beta}{\beta+1}$ where $\beta = \frac{A_S}{2}$. Combining with Equation~\ref{eq:1d} gives
$$ \Big(\frac{\beta}{\beta+1}\Big)^{1-r} \leq \Big(\frac{|S|}{|S|+1}\Big)^{1-r} \leq \frac{A_S}{A_S + 2(1-r)} $$
Rearranging and setting $1-r=\frac{1}{t}$ implies
$$ \frac{A_S+2(1-r)}{A_S} \leq \Big(\frac{\beta+1}{\beta}\Big)^{\frac{1}{t}} \Longrightarrow \Big(1+ \frac{2(1-r)}{A_S}\Big)^{t} \leq 1+ \frac{1}{\beta}$$
Bernoulli's inequality states if $x,q\in \mathbb{R}$ and $x>-1, q>1$ then $(1+x)^{q}>1+qx$. Applying the inequality for
$x=\frac{2(1-r)}{A_S}$ and $q=t=\frac{1}{1-r}>1$ gives
$$ 1+\frac{1}{\beta} = 1+ \frac{2t(1-r)}{A_S} < \Big(1+ \frac{2(1-r)}{A_S}\Big)^{t} \leq 1+\frac{1}{\beta}$$
which is a contradiction.
\end{proof}

\begin{definition}
Given a partition of the people into groups, we say a group $G$ is a \textbf{\emph{richest}} group if its expected utility per
person value is at least that of any other group in the partition.
\end{definition}

We note that given any partition of people into groups, there always exists at least one richest group. We show the following
interesting and unexpected phenomenon: in every Nash Equilibrium any richest group covers a constant fraction of the total
volume.

\begin{theorem}
For the line case, in any Nash Equilibrium there is always a group which covers a constant fraction of the total length where
the constant is $\frac{1}{1+2(1-r)}$. \label{thm:constant-ratio-for-1d}
\end{theorem}
\begin{proof}
We refer to Figure~\ref{fig:constant} for the notation used in this proof. Given a Nash Equilibrium, consider any richest
group $S$. Since we contracted the points not covered by anybody, a player $x\notin S$ has coverage length starting from the
leftmost point of $E_L$ or $E_L$ might be empty. If $E_L$ is not empty then $|E_L|\leq 2(1-r)$ else $x$ contradicts
Lemma~\ref{lemma:no-one-has-large-exclusive-area}. In either case $|E_L|\leq 2(1-r)$. Similarly there is a player $y\notin S$
whose coverage length ends at the rightmost point of $E_R$ or $E_R$ is empty. In either case $|E_R|\leq 2(1-r)$. Claim is
$|I_j|\leq 4(1-r)$ for $1\leq j\leq m-1$. Suppose $\exists\ j$ such that $|I_j| > 4(1-r)$. The midpoint of $I_j$ must be
covered by a player $z\notin S$ and so this $z$ covers at least half of $I_j$, i.e., $z$ can offer greater than $2(1-r)$
length to $S$ which contradicts Lemma~\ref{lemma:no-one-has-large-exclusive-area}. Also $|S_j|\geq 2$ for all $1\leq j\leq m$
as each $S_j$ is a concatenation of the coverage lengths of one or more members of $S$. Therefore $|A_S| = \sum_{j=1}^{m}
|S_j| \geq 2m$ and
\begin{align*}
 |A| &= \sum_{j=1}^{m}|S_j| + \sum_{j=1}^{m-1}|I_j| + |E_L| + |E_R| \\
 &\leq \sum_{j=1}^{m}|S_j| + \Big(4(m-1)(1-r)\Big) + 2(1-r) + 2(1-r) \\
 &\leq \sum_{j=1}^{m}|S_j| + 4(1-r)m \\
 &\leq \sum_{j=1}^{m}|S_j| + 2(1-r)\sum_{j=1}^{m}|S_j| \\
 &= \Big(1+2(1-r)\Big)\Big(\sum_{j=1}^{m} |S_j|\Big)
\end{align*}
Thus $\dfrac{|A_S|}{|A|} = \dfrac{\sum_{j=1}^{m} |S_j|}{|A|}\geq \dfrac{1}{1+2(1-r)}$, i.e., the group $S$ covers a constant
fraction of the total length.
\end{proof}

\vspace{10mm}

\begin{figure}[h]
\centering
\includegraphics[width=4in]{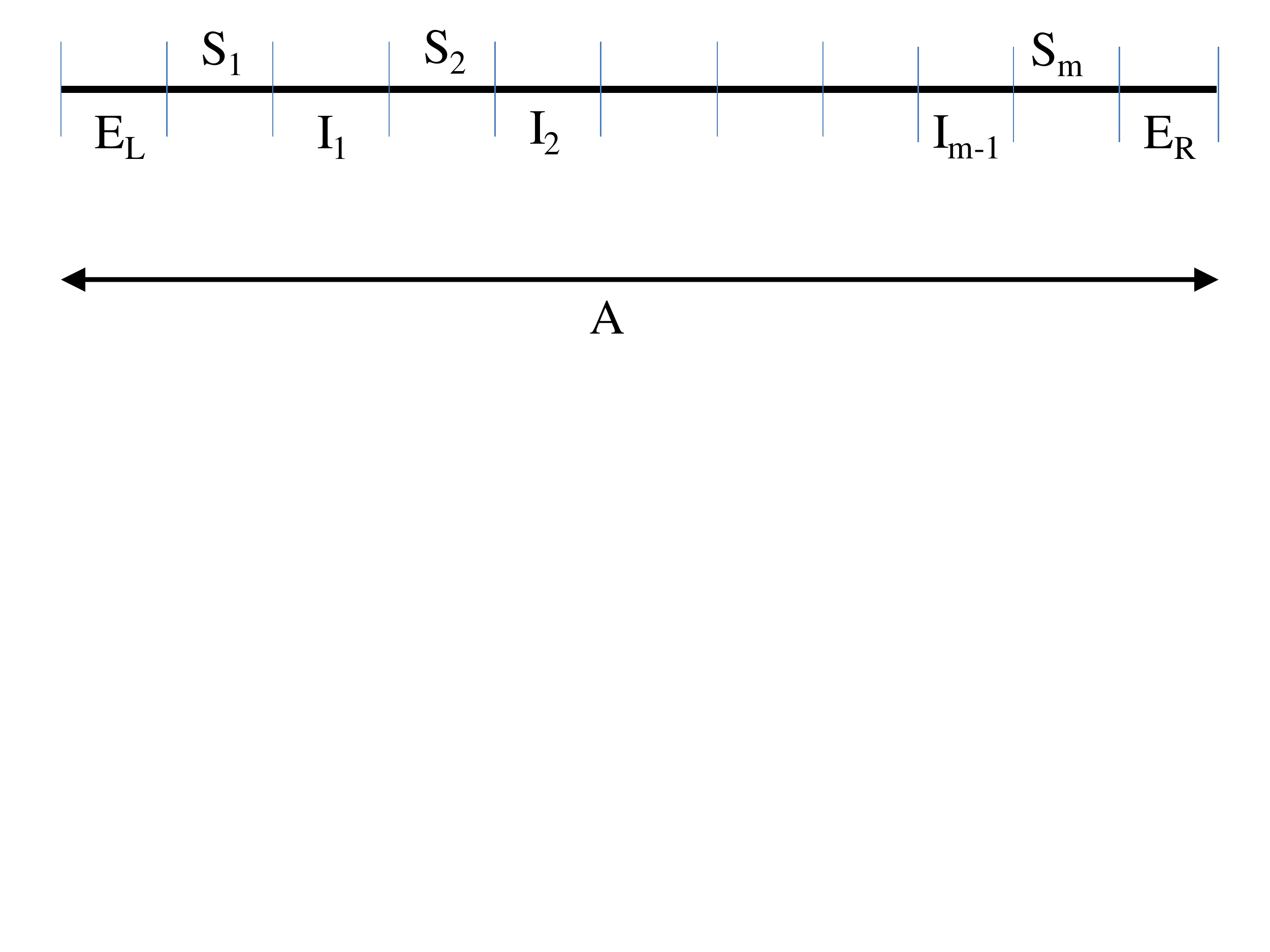}
\vspace{-40mm}
\caption
{
The segments covered by $S$ are $S_1,S_2,\ldots,S_m$.
The internal gaps are $I_1,I_2,\ldots,I_{m-1}$. The left and right
external gaps are $E_L$ and $E_R$ respectively. The total length $A$ stretches from the left endpoint of $E_L$ to the right endpoint of $E_R$
\label{fig:constant}
}
\end{figure}

\section{The Euclidean $d$-dimensional Case}
\label{sec:euclidean-d}

In this section we consider the case in which the players are located in a Euclidean $d$-dimensional space and each person
covers a unit ball around himself. The next lemma bounds the ratio of volumes of the union of the two families of balls with
the same set of centers but different radii.
\begin{lemma}
\label{lem:bound-on-union} Let $A$ be a finite family of balls of radius one in a Euclidean $d$-dimensional space. Let $B$ be
a family of balls with the same set of the centers but radius $t\geq 1$. Let $A_U,B_U$ denote the union of balls in $A$ and
$B$ respectively. Then Vol$(B_U)\leq t^d\cdot$ Vol$(A_U)$ where Vol$(A_U),$Vol$(B_U)$ denotes the volume of $A_U$ and $B_U$
respectively.
\end{lemma}
\begin{proof}
Let $C=\{c_1,c_2,\ldots,c_n\}$ be the centers of the balls in $A$. For $x\in B_U$ define c$(x) = $min$\{\ d(c_j,x)\ |\ 1\leq
j\leq n\ \}$. Consider the partition $P_1,P_2,\ldots,P_n$ of $B_U$ into $n$ parts: $x\in B_U$ is placed in $P_i$ if and only
if $i = $ min$\{\ j\ |\ d(x,c_j) = c(x)\ \}$.

We claim $y\in P_j$ implies $[c_j,y]\in P_j$. Suppose there is a point $z\in [c_j,y]$ such that $z\in P_k$ for $k\neq j$. By
the triangle inequality $d(c_k,y)\leq d(c_k,z)+d(z,y)\leq d(c_j,z)+d(z,y) = d(c_j,y)$ where we used $z\in P_k$ implies
$d(c_k,z)\leq d(c_j,z)$. So $d(c_k,y)\leq d(c_j,y)$. But $y\in P_j$ implies $d(c_k,y) = d(c_j,y),\ d(c_k,z) = d(c_j,z)$ and
$j<k$. This contradicts the membership of $z$ in $P_k$. So we can apply homothecy: for every $1\leq i\leq n$ we contract each
$P_i$ w.r.t point $c_i$ by a factor of $\frac{1}{t}$ to get a region say $P'_i$. We note $P'_i\subseteq A_U$ as $x\in P_i$
implies $d(x,c_i)\leq t$ and if we denote by $x'$ the point to which $x$ is mapped under the contraction, then $d(x',c_i) =
\frac{1}{t}\cdot d(x,c_i) \leq \frac{1}{t}\cdot t = 1$.

The next claim is $P'_i\cap P'_j = \emptyset$ for any $i\neq j$. Suppose not and say $y\in P'_i\cap P'_j$. Let $y_i,y_j$ be
the points in $P_i$ and $P_j$ respectively which get mapped to $y$ under the contraction. Let $d(c_j,y)=\alpha$ and
$d(c_i,y)=\beta$. By the triangle inequality we have $d(c_i,y_j) \leq d(c_i,y) + d(y,y_j) = \beta +(t-1)\alpha$. Also $y_j\in
P_j$ implies $d(c_i,y_j)\geq d(c_j,y_j) = t\alpha$. So $\beta +(t-1)\alpha \geq d(c_i,y_j)\geq d(c_j,y_j)\geq t\alpha$, i.e.,
$\beta \geq \alpha$. Similarly we have $\alpha \geq \beta$ which implies $\alpha = \beta$. Therefore $t\alpha = d(c_j,y_j)
\leq d(c_i,y_j) \leq d(c_i,y) + d(y,y_j) = \beta +(t-1)\alpha = t\alpha$. Equality in the triangle inequality gives
$c_i,y_j,y_i$ are on the same line and $d(c_i,y_j)=d(c_i,y_i)$ which implies $y_i=y_j$ which is a contradiction. So we have
the following two conditions :
\begin{enumerate}
\item $P'_i\subseteq A_U$ for every $1\leq i\leq n$.
\item $P'_i\cap P'_j = \emptyset$ for any $i\neq j$.
\end{enumerate}
Therefore, Vol$(A_U) \geq \sum_{i=1}^{n}$ Vol$(P'_i) = \frac{1}{t^d}\cdot \sum_{i=1}^{n}$ Vol$(P_i)= \frac{1}{t^d}\cdot$
Vol$(B_U)$. We note the bound is tight when all the balls in $B$ are disjoint.
\end{proof}

We now prove a generalization of Lemma~\ref{lemma:no-one-has-large-exclusive-area} for a Euclidean $d$-dimensional space. Let
$V_d$ denote the volume of a unit ball in the Euclidean $d$-dimensional space.

\begin{lemma}
Let $S$ be a richest group in a Nash Equilibrium. There is no player $i \notin S$ who can add a volume of at least $(1-r)V_d$
to the volume $A_S$ currently covered by $S$. \label{lem:no-one-large-d-exlusive-area}
\end{lemma}
\begin{proof}
Suppose there is a player $i \notin S$, who can add a volume of at least $(1-r)V_d$ to the volume covered by $S$. However,
since it is a Nash Equilibrium, either the new expected utility of $S$ on adding this player is less than or equal to the
current expected utility of $S$, hence $S$ would have no incentive in adding the player $i$. Or else the player $i$ would not
have any incentive to move to $S$, as the projected new expected utility of $i$ is less than or equal to his current expected
utility. Since $S$ is a richest group, both these conditions combine to give:
\begin{equation}
 \Big(\frac{M}{|S|+1}\Big)^{1-r}\cdot \frac{A_S + (1-r)V_d}{A} \leq \Big(\frac{M}{|S|}\Big)^{1-r}\cdot \frac{A_S}{A}
\label{eq:d-dimension}
\end{equation}
As each player has a coverage volume of $V_d$ we have $|S|V_d \geq A_S$ (with equality only if the coverage volumes of the
members of $S$ are pairwise disjoint). The function $f(x)=\frac{x}{x+1}$ is increasing on $(0,\infty)$ and hence
$\frac{|S|}{|S|+1} \geq \frac{\beta}{\beta+1}$ where $\beta = \frac{A_S}{V_d}$. Combining with Equation~\ref{eq:d-dimension}
gives
$$ \Big(\frac{\beta}{\beta+1}\Big)^{1-r} \leq \Big(\frac{|S|}{|S|+1}\Big)^{1-r} \leq \frac{A_S}{A_S + (1-r)V_d} $$
Rearranging and setting $1-r=\frac{1}{t}$ implies
$$ \frac{A_S+(1-r)V_d}{A_S} \leq \Big(\frac{\beta+1}{\beta}\Big)^{\frac{1}{t}} \Longrightarrow \Big(1+ \frac{(1-r)V_d}{A_S}\Big)^{t} \leq 1+ \frac{1}{\beta}$$
Bernoulli's inequality states if $x,q\in \mathbb{R}$ and $x>-1, q>1$ then $(1+x)^{q}>1+qx$. Applying the inequality for
$x=\frac{(1-r)V_d}{A_S}$ and $q=t=\frac{1}{1-r}>1$ gives
$$ 1+\frac{1}{\beta} = 1+ \frac{t(1-r)V_d}{A_S} < \Big(1+ \frac{(1-r)V_d}{A_S}\Big)^{t} \leq 1+ \frac{1}{\beta}$$
which is a contradiction.
\end{proof}

We first give a simple proof that if $S$ is any richest group in a Nash Equilibrium then $\frac{A_S}{A}\geq \frac{1}{3^d}$
where $A_S$ is the volume covered by the group $S$ and $A$ is the total volume.

\begin{theorem}
\label{thm:atleast-3^d} If the players are located in a $d$-dimensional Euclidean space then in any Nash Equilibrium every
richest group covers at least a $\frac{1}{3^d}$-fraction of the total volume.
\end{theorem}
\begin{proof}
Let $S$ be a richest group in a Nash Equilibrium and consider a player $x\notin S$. There must exist a member of $s\in S$ such
that the distance between $x$ and $s$ is at most two. Otherwise the coverage ball of $x$ is disjoint from the coverage ball of
$S$ and thus can contribute a volume of $V_d > (1-r)V_d$ which contradicts Lemma~\ref{lem:no-one-large-d-exlusive-area}. So
the total volume $A$ is covered by the volume $A'_{S}$ of the union of the family of balls of radius three centered at the
members of $S$. By Lemma~\ref{lem:bound-on-union} we have $\dfrac{A_S}{A}\geq \dfrac{A_S}{A'_S}\geq \dfrac{1}{3^d}$.
\end{proof}

The above bound of $\frac{1}{3^d}$ is independent of the risk aversion factor $r$. We now give a better bound which depends on
$r$. To this end, in the following lemma, we show how to bound the volume of intersection of two unit balls in a Euclidean
$d$-dimensional space in terms of the distance between their centers.

\begin{lemma}
Let $B_1$ and $B_2$ be two unit balls in a Euclidean $d-$dimensional space. For $a\leq 1$, if the distance between the centers
of $B_1$ and $B_2$ is 2$a$, then the volume of intersection of $B_1$ and $B_2$ is at most $2(1-a^2)^{\frac{d-1}{2}}\cdot
V_{d-1}$ where $V_{d-1}$ is the volume of a unit ball in a Euclidean $(d-1)$-dimensional space. \label{lem:bound-intersection}
\end{lemma}
\begin{proof}
We assume $a\leq 1$ otherwise the area of intersection is clearly zero. Let $c_1,c_2$ be the centers of $B_1$ and $B_2$
respectively. Then $B_1\cap B_2$ is contained in the cylinder (say $B$) of radius $\sqrt{1-a^2}$ and height $2a$ centered at
the midpoint of the segment joining $c_1$ and $c_2$. So Vol$(B_1\cap B_2)\leq $ Vol$(B) = 2a(1-a^2)^{\frac{d-1}{2}}V_{d-1}
\leq 2(1-a^2)^{\frac{d-1}{2}}\cdot V_{d-1}$ as $a \leq 1$.
\end{proof}

The next lemma gives a lower bound on the ratio of the volumes of unit balls in Euclidean spaces of consecutive dimensions.

\begin{lemma}
For $d\geq 2$ let $V_{d-1},V_d$ be the volumes of unit balls in the $d-1$ and $d-$dimensional Euclidean spaces respectively.
Then $\frac{V_d}{V_{d-1}}\geq \frac{1}{d}$. \label{lem:consecutive-dimension}
\end{lemma}
\begin{proof}
We prove by induction on $d$. Base case is d=2 and $\frac{V_2}{V_1} = \frac{\pi}{2}\geq \frac{1}{2}$. We use the well-known
recurrence relation for $V_n$: $V_n = \frac{2\pi}{n}\cdot V_{n-2}$. Suppose the hypothesis is true for all $k\leq n-1$. Then
we have $\frac{V_n}{V_{n-1}} = \frac{n-1}{n} \cdot \frac{V_{n-2}}{V_{n-3}}\geq \frac{n-1}{n} \cdot \frac{1}{n-2} >
\frac{1}{n}$ and so the hypothesis holds true for all $d\geq 2$.
\end{proof}

We are now ready to give a better bound than $\frac{1}{3^d}$ on the fraction of the total volume covered by any richest group
in a Nash Equilibrium.

\begin{theorem}
If the players are located in a $d$-dimensional Euclidean space, then in any Nash Equilibrium there always is a group which
covers at least a $\frac{1}{(2\delta+1)^d}$-fraction of the total volume where $\delta =
\sqrt{1-(\frac{r}{2d})^{\frac{2}{d-1}}}$. We note $\frac{1}{(2\delta+1)^d} > \frac{1}{3^d}$  as $\delta < 1$ and therefore
this improves on Theorem~\ref{thm:atleast-3^d}. \label{thm:better-than-3^d}
\end{theorem}
\begin{proof}
Consider a richest group $S$ in a Nash Equilibrium. By Lemma~\ref{lem:no-one-large-d-exlusive-area}, no player outside of $S$
can get his coverage ball to contribute at least $(1-r)V_d$ volume to $S$, i.e, for every $x\notin S$ there is a player $s\in
S$ such that volume of intersection of balls $B_x,B_s$ of $x$ and $s$ respectively is at least $rV_d$. Let the distance
between centers of $B_x$ and $B_s$ be $2a$. Lemma~\ref{lem:bound-intersection} gives $2(1-a^2)^{\frac{d-1}{2}}V_{d-1}\geq$
Vol$(B_s\cap B_x)\geq r\cdot V_d$ which implies $2(1-a^2)^{\frac{d-1}{2}}\geq r\frac{V_d}{V_{d-1}}\geq \frac{r}{d}$ by
Lemma~\ref{lem:consecutive-dimension}. Rearranging we get $a\leq \sqrt{1-(\frac{r}{2d})^{\frac{2}{d-1}}} =$ say $\delta$. So
each player not in $S$ is at a distance of at most $2\delta$ from some player of $S$. Therefore the total volume $A$ is
covered by the volume $A'_{S}$ of the union of the family of balls of radius $2\delta+1$ centered at members of $S$. By
Lemma~\ref{lem:bound-on-union} we have $\frac{A_S}{A}\geq \frac{A_S}{A'_S}\geq \frac{1}{(2\delta+1)^d}$.
\end{proof}

\section{The Graph Case}
\label{sec:graph-case}

In this section we consider the discrete version of the problem where players form the vertex set of an undirected graph. The
coverage of a vertex is its closed neighborhood, i.e., a vertex covers itself and all its neighbors. We assume the same
utility function as before: Each member $x$ belonging to a group $S$ has expected utility given by $E[u(x)]=
\Big(\frac{M}{|S|}\Big)^{1-r}\cdot \frac{|A_S|}{|A|}$ where $M$ is the total money, $A_S$ is the union of the closed
neighborhoods of the vertices in $S$ and $A$ is the vertex set of the graph. We first show a preliminary lemma which bounds
the contribution to a richest group in a Nash Equilibrium by any vertex which is not in the richest group. This lemma can be
viewed as a discrete version of Lemma~\ref{lem:no-one-large-d-exlusive-area}.

\begin{lemma}
Let $G=(V,E)$ be an undirected graph with maximum degree $\Delta$. Let $S$ be a richest group in a Nash Equilibrium. Then
there is no player $i \notin A_S$ who can add at least $(1-r)(\Delta + 1)$ vertices to the set $A_S$ currently covered by $S$.
\label{lem:bound-max-degree-richest}
\end{lemma}
\begin{proof}
Suppose a player $i \notin A_S$ can add at least $(1-r)(\Delta + 1)$ vertices to the set $A_S$ currently covered by $S$.
However, since it is a Nash Equilibrium, either the new expected utility of $S$ on adding this player is less than or equal to
the current expected utility of $S$, hence $S$ would have no incentive in adding the player $i$. Or else the player $i$ would
not have any incentive to move to $S$, as the projected new expected utility of $i$ is less than or equal to its current
expected utility. Since $S$ is a richest group, both these conditions combine to give:
\begin{equation}
 \Big(\frac{M}{|S|+1}\Big)^{1-r}\cdot \frac{A_S + (1-r)(\Delta + 1)}{A} \leq \Big(\frac{M}{|S|}\Big)^{1-r}\cdot \frac{A_S}{A}
\label{eq:graph-case}
\end{equation}
As each player has degree at most $\Delta$ we have $|S|(\Delta + 1) \geq A_S$ (with equality only if the closed neighborhoods
of the vertices of $S$ are pairwise disjoint). Since $f(x)=\frac{x}{x+1}$ is an increasing function we have $\frac{|S|}{|S|+1}
\geq \frac{\beta}{\beta+1}$ where $\beta = \frac{A_S}{\Delta + 1}$. Combining with Equation~\ref{eq:graph-case} gives
$$ \Big(\frac{\beta}{\beta+1}\Big)^{1-r} \leq \Big(\frac{|S|}{|S|+1}\Big)^{1-r} \leq \frac{A_S}{A_S + (1-r)(\Delta + 1)} $$
Rearranging and setting $1-r=\frac{1}{t}$ implies
$$ \frac{A_S+(1-r)(\Delta + 1)}{A_S} \leq \Big(\frac{\beta+1}{\beta}\Big)^{\frac{1}{t}} \Longrightarrow
\Big(1+ \frac{(1-r)(\Delta + 1)}{A_S}\Big)^{t} \leq 1+ \frac{1}{\beta}$$
Bernoulli's inequality states if $x,q\in \mathbb{R}$ and $x>-1, q>1$ then $(1+x)^{q} > 1+qx$. Applying the inequality for
$x=\frac{(1-r)(\Delta + 1)}{A_S}$ and $q=t=\frac{1}{1-r}>1$ gives
$$ 1+\frac{1}{\beta} = 1+ \frac{t(1-r)(\Delta + 1)}{A_S} < \Big(1+ \frac{(1-r)(\Delta + 1)}{A_S}\Big)^{t} \leq 1+ \frac{1}{\beta}$$
which is a contradiction.
\end{proof}

In the next theorem we show if the topology is the class of bounded-degree regular graphs, then in any Nash Equilibrium there
always exists a group which covers a constant fraction of the total number of vertices.

\begin{theorem}
Let $G=(A,E)$ be a $f$-regular graph. In any Nash Equilibrium there always exists a group covering a constant fraction of the
total number of vertices where the constant is $\dfrac{1}{\frac{f-1}{r(f+1)}+1}$. \label{thm:regular-graphs}
\end{theorem}
\begin{proof}
Consider a richest group $S$ in a Nash Equilibrium. Let $A_S$ be the vertices covered by $S$, i.e, $A_S$ is the union of the
closed neighborhoods of the vertices of $S$. Denote by $\overline{A_S}$ the set $A\setminus A_S$. Since the graph is
$f$-regular the size of the closed neighborhood of every vertex is $f+1$. By Lemma~\ref{lem:bound-max-degree-richest}, every
vertex $x\notin A_S$ must add less than $(1-r)(f+1)$ vertices to the set $A_S$. So each vertex in $\overline{A_S}$ has at
least $(f+1)-(1-r)(f+1) = r(f+1)$ neighbors in $A_S$. Let $\beta$ be the number of edges with one endpoint in $A_S$ and one
endpoint in $\overline{A_S}$. Thus $\beta \geq r|\overline{A_S}|(f+1)$. By the definition of $A_S$ as the union of the closed
neighborhoods of vertices of $S$, only the vertices from $A_S\setminus S$ can have edges to $\overline{A_S}$. Each vertex of
$A_S\setminus S$ has at least one neighbor in $S$ and so $\beta \leq |A_S|(f-1)$. Combining the two bounds we have
$|A_S|(f-1)\geq \beta \geq r|\overline{A_S}|(f+1) = r(|A|-|A_S|)(f+1)$. Letting $\mu = \frac{f-1}{r(f+1)}$ we have $\mu
|A_S|\geq |A| - |A_S|$, i.e., $\dfrac{|A_S|}{|A|}\geq \dfrac{1}{\mu + 1} =\dfrac{1}{\frac{f-1}{r(f+1)}+1}$
\end{proof}

The general graph case does not seem to be hopeful. Recall in all the three topologies (the one-dimensional (line) space, the
$d$-dimensional Euclidean space and the bounded-degree regular graphs) considered so far, we were able to show the surprising
phenomenon that any \emph{richest} group in a Nash Equilibrium covers a constant fraction of the total volume/vertices. We
show this approach fails for general graphs, i.e., there exist graphs having a Nash Equilibrium in which no richest group
covers a constant fraction of the total number of vertices.

\begin{theorem}
There exist graphs which have a Nash Equilibrium in which no richest group covers a constant fraction of the total number of
vertices. \label{thm:richest-not-always-constant-fraction}
\end{theorem}
\begin{proof}
Consider the family of graphs $G_z$ where $z$ is a parameter satisfying the equation
\begin{equation}
\Big(\dfrac{M}{1}\Big)^{1-r}\cdot \dfrac{12z^r}{|G_z|} > \Big(\dfrac{M}{2}\Big)^{1-r}\cdot \dfrac{12z^r + 4}{|G_z|}
\hspace{3mm} , i.e., \hspace{3mm} z^r> \dfrac{1}{3(2^{1-r}-1)}
\label{eq:lower-bound-on-z}
\end{equation}
We now describe the graph $G_z$: it contains a clique $K$ of size $12z^r$. We say these vertices are of Type I. A path $P$ of
length $3z$ is attached to a vertex (say $v$) of Type I. We call the vertices of the path $P$ (excluding $v$) as the vertices
of Type II.

First we show no two vertices $x,y$ of Type I merge. If $v\notin \{x,y\}$ then $x$ and $y$ will not merge as they both have
the same coverage. So without loss of generality let $x=v$. Then the initial expected utility of $y$ is
$\Big(\dfrac{M}{1}\Big)^{1-r}\cdot \dfrac{12z^r}{|G_z|}$ and the expected utility of the group $\{x,y\}$ is
$\Big(\dfrac{M}{2}\Big)^{1-r}\cdot \dfrac{12z^r + 1}{|G_z|}$. Equation~\ref{eq:lower-bound-on-z} implies
$\Big(\dfrac{M}{1}\Big)^{1-r}\cdot \dfrac{12z^r}{|G_z|} > \Big(\dfrac{M}{2}\Big)^{1-r}\cdot \dfrac{12z^r + 4}{|G_z|} >
\Big(\dfrac{M}{2}\Big)^{1-r}\cdot \dfrac{12z^r + 1}{|G_z|}$. Hence no two vertices of Type I merge.

We next show no vertex $p$ of Type I merges with a vertex $q$ of Type II. The initial expected utility of $p$ is at least
$\Big(\dfrac{M}{1}\Big)^{1-r}\cdot \dfrac{12z^r}{|G_z|}$ with equality if $p\neq v$. The expected utility of the group
$\{p,q\}$ will be at most $\Big(\dfrac{M}{2}\Big)^{1-r}\cdot \dfrac{12z^r + 4}{|G_z|}$ as any vertex of Type II can add at
most 3 vertices to coverage of a vertex of Type I. Equation~\ref{eq:lower-bound-on-z} implies no vertex of Type I will merge
with a vertex of Type II.

So a group in any Nash Equilibrium has to either be a single vertex of Type I or a set of Type II vertices. In the first case
the maximum expected utility will be for the group formed by $v$ alone  and is given by
\begin{equation}
\Big(\frac{M}{1}\Big)^{1-r}\cdot \frac{12z^r + 1}{|G_z|} = \hspace{1mm}say \hspace{2mm}U_1
\label{eq:first-kind-of-group}
\end{equation}
For the second case if the group consist of $b$ vertices of Type II then the expected utility of this group is
$\Big(\frac{M}{b}\Big)^{1-r}\cdot \frac{C}{|G_z|} = $say $U_2$ where $C$ is the coverage of the group. The coverage of any
vertex of Type II is at most three implies $C\leq 3b$ with equality only if the coverages of the members of the group are
pairwise disjoint. So $U_2 = \Big(\frac{M}{b}\Big)^{1-r}\cdot \frac{C}{|G_z|} \leq \Big(\frac{M}{b}\Big)^{1-r}\cdot
\frac{3b}{|G_z|} = (M)^{1-r}\cdot \frac{3b^{r}}{|G_z|} \leq (M)^{1-r}\cdot \frac{3(3z)^{r}}{|G_z|} <
\Big(\frac{M}{1}\Big)^{1-r}\cdot \frac{12z^r + 1}{|G_z|} = U_1$ as $b\leq 3z$ and $r\in (0,1)$. Therefore the only richest
group in any Nash Equilibrium in $G_z$ is the group $\{v\}$. The fraction of vertices covered by v is $\frac{12z^r + 1}{12z^r
+ 3z} < \frac{12z^r + 3z^{r}(2^{1-r}-1)}{12z^r + 3z} = \frac{4 + (2^{1-r}-1)}{4 + z^{1-r}} = \frac{3+2^{1-r}}{3 + z^{1-r}}$
which tends to 0 as we increase $z$ since $r\in (0,1)$ (Equation~\ref{eq:lower-bound-on-z} only imposed a lower bound on $z$
and hence there is no issue with increasing $z$ arbitrarily). So there exist graphs which have a Nash Equilibrium in which no
richest group covers a constant fraction of the total number of vertices.
\end{proof}

Theorem~\ref{thm:richest-not-always-constant-fraction} implies we need different techniques than the ones used above to
resolve the general graph case. However, under the assumption that \emph{defection to an empty group} is not allowed, we can
show given any constant $c < 1$ there exists a graph $G_c$ and a Nash Equilibrium in $G_c$ such that each group in the Nash
Equilibrium covers strictly less than a $c$-fraction of the total number of vertices.  We now explicitly construct such
graphs. Consider the family of graphs $G_{k,\ell}$: it has a clique of size $k$ formed by the vertices
$\{v_1,v_2,\ldots,v_k\}$. We call these vertices as the \textbf{primary} vertices. Each primary vertex $v_{i}$ has $\ell$
leaves attached to it. We denote these \textbf{secondary} vertices attached to the primary vertex $v_i$ by $L(v_i) =
\{v_{i,1},v_{i,2},\ldots,v_{i,\ell}\}$. We note $|G_{k,\ell}| = k+k\ell$.


\begin{lemma}
\label{lem:Nash-graph-case} If $k$ and $\ell$ satisfy
\begin{equation}
\label{eq:k-and-l-condition}
k > \ell \Big(\frac{2-f(\ell)}{f(\ell)-1}\Big) \hspace{3mm} where \hspace{3mm} f(\ell) = \Big(1 + \frac{1}{\ell + 1}\Big)^{1-r}
\end{equation}
then the groups $S_1,S_2,\ldots,S_k$ defined by $S_i = L(v_i)\cup \{v_i\}$ form a Nash Equilibrium for the graph $G_{k,\ell}$.
\end{lemma}
\begin{proof}
The function $g(x)= \frac{2-x}{x-1}$ is decreasing in the interval $(1,2^{1-r})$. Also $f(\ell) < 2^{1-r}$ implies
\begin{equation}
\label{eq:gamma}
g(2^{1-r}) < g(f(\ell))\hspace{2mm} ,i.e.,\hspace{2mm} g(2^{1-r}) = \frac{2-2^{1-r}}{2^{1-r}-1} < \frac{2-f(\ell)}{f(\ell)-1} < \frac{k}{\ell}
\end{equation}
where the last step follows from Equation~\ref{eq:k-and-l-condition}. The quantity $g(2^{1-r})$ is a constant and we denote it
by say $\gamma$.

We show no two of the groups $S_1,S_2,\ldots,S_k$ will merge. The current expected utility of any $S_i$ is $\Big(\frac{M}{\ell
+ 1}\Big)^{1-r}\cdot \frac{k + \ell}{k + k\ell} =$ say $u_1$ and if any $S_i$ and $S_j$ merge then the expected utility of the
new merged group will be $\Big(\frac{M}{2(\ell + 1)}\Big)^{1-r}\cdot \frac{k + 2\ell}{k + k\ell} =$ say $u_2$. Now
\begin{align*}
 u_1 > u_2 &\Leftrightarrow \Big(\frac{M}{\ell + 1}\Big)^{1-r}\cdot \frac{k + \ell}{k + k\ell} > \Big(\frac{M}{2(\ell + 1)}\Big)^{1-r}\cdot \frac{k + 2\ell}{k + k\ell}\\
 &\Leftrightarrow 2^{1-r}(k + \ell) > (k + 2\ell)\\
 &\Leftrightarrow k > \ell \Big(\frac{2-2^{1-r}}{2^{1-r}-1}\Big) = \ell \cdot g(2^{1-r})
\end{align*}
which follows from Equation~\ref{eq:gamma}. So no two of the groups $S_1,S_2,\ldots,S_k$ will merge.

We next show no player will defect from one group to another. The number of players in the new group will be $\ell + 2$ and it
will have more coverage if a primary vertex defects rather than a secondary vertex. So it is enough to prove the defection of
a primary vertex to another group is not possible. Suppose the primary vertex $v_i$ defects to join the group $S_j$. The
initial expected utility of $v_i$ is $\Big(\frac{M}{\ell + 1}\Big)^{1-r}\cdot \frac{k + \ell}{k + k\ell}=$ say $u_3$ and the
expected utility of the group $\{v_i\}\cup S_j$ is $\Big(\frac{M}{\ell + 2}\Big)^{1-r}\cdot \frac{k + 2\ell}{k + k\ell}=$ say
$u_4$. Now
\begin{align*}
 u_3 > u_4 &\Leftrightarrow \Big(\frac{M}{\ell + 1}\Big)^{1-r}\cdot \frac{k + \ell}{k + k\ell} > \Big(\frac{M}{\ell + 2}\Big)^{1-r}\cdot \frac{k + 2\ell}{k + k\ell}\\
 &\Leftrightarrow \Big(\frac{\ell + 2}{\ell + 1}\Big)^{1-r}\cdot (k + \ell) > (k + 2\ell)\\
 &\Leftrightarrow f(\ell) \cdot (k + \ell) > (k + 2\ell) \\
 &\Leftrightarrow k > \ell \Big(\frac{2-f(\ell)}{f(\ell)-1}\Big)
\end{align*}
which holds by Equation~\ref{eq:k-and-l-condition}. Therefore no defection will take place. Since no merging or defection can
occur, the groups $S_1,S_2,\ldots,S_k$ form a Nash Equilibrium in the graph $G_{k,\ell}$.
\end{proof}

Equations~\ref{eq:k-and-l-condition} and ~\ref{eq:gamma} (given in proof of Lemma~\ref{lem:Nash-graph-case}) do not impose any
absolute upper bounds on $k$ or $\ell$. We use this fact in the following theorem to show for every positive constant $c < 1$
there exists a graph $G_c$ and a Nash Equilibrium in $G_c$ such that no group in the Nash Equilibrium covers at least a
$c$-fraction of the total number of vertices.

\begin{theorem}
\label{thm:graph-no-constant} Given any positive constant $c < 1$, under the assumption that defecting to an empty group is
not allowed, there exists a graph $G_c$ and a Nash Equilibrium in $G_c$ such that each group in the Nash Equilibrium covers
strictly less than a $c$-fraction of the total number of vertices.
\end{theorem}
\begin{proof}
We set $G_c$ to be the graph $G_{k,\ell}$ where $k,\ell$ will be determined later. Consider the Nash Equilibrium in $G_c$
given by the groups $S_1,S_2,\ldots,S_k$ in Lemma~\ref{lem:Nash-graph-case}. The fraction of the total number of vertices
covered by any group $S_i$ is given by $\dfrac{k + \ell}{k + k\ell}  = \dfrac{1 + \frac{\ell}{k}}{1 + \ell} < \dfrac{1 +
\frac{1}{\gamma}}{1 + \ell}$ where the last step follows from Equation~\ref{eq:gamma}. Recalling $\gamma = g(2^{1-r}) =
\dfrac{2-2^{1-r}}{2^{1-r}-1}$ is a constant, we choose $\ell$ large enough so $c > \dfrac{1 + \frac{1}{\gamma}}{1 + \ell}$
which proves our theorem. We need to choose $k$ large enough to satisfy Equation~\ref{eq:k-and-l-condition} but this is not an
issue as we do not have any absolute upper bound constraints on either $k$ or $\ell$.
\end{proof}

\section{Conclusions and Open Problems}


In this paper we have suggested a game-theoretic model motivated by the DARPA Network Challenge. We analyze the structures of
the groups in Nash equilibria. We show for various topologies: a one-dimensional space (line), a $d$-dimensional Euclidean
space, and bounded-degree regular graphs; in any Nash Equilibrium there always exists a group which covers a constant fraction
of the total volume. The objective of events like the DARPA Network Challenge is to mobilize a large number of people quickly
so that they can cover a big fraction of the total area. Our results suggest that this objective can be met under certain
conditions.

However our ideas however do not generalize to all the graphs and we provide explicit examples of graphs for which our
techniques fail. Under an additional assumption that defecting to an empty group is not allowed, we show given any constant $c
< 1$ there exists a graph $G_c$ and a Nash Equilibrium in $G_c$ where each group in the Nash Equilibrium covers strictly less
than a $c$-fraction of the total number of vertices. It would be interesting to prove Theorem~\ref{thm:graph-no-constant}
without the assumption that defecting to an empty group is prohibited.


\bibliography{ref}

\appendix

\section{Omitted Proofs from Section~\ref{sec:results-and-techniques}}

\paragraph{\textbf{Proof of Theorem~\ref{thm:increasing-area}}}

\begin{proof}
Let $S$ be a group which covers volume $A_S = \lambda A$, where $A$ is the total volume. Since each member of $S$ contributes
volume~1 we have $|S|\geq A_S \geq \lambda A \geq \lambda^{\frac{1}{1-r}+1}A$ (because $\lambda \in (0,1)$). If each player in
$S$ has expected utility at least~$c$ then $c\leq E[u(x)] \leq \lambda \cdot \Big(\frac{M}{|S|}\Big)^{1-r}$, which implies
$\Big(\frac{M}{|S|}\Big)^{1-r}\geq \frac{c}{\lambda}$. So $M\geq |S|(\frac{c}{\lambda})^{\frac{1}{1-r}}\geq
\lambda^{\frac{1}{1-r}+1}A(\frac{c}{\lambda})^{\frac{1}{1-r}} = \lambda A c^{\frac{1}{1-r}}$, proving the lemma.
\end{proof}

\paragraph{\textbf{Proof of Theorem~\ref{thm:increasing-size}}}

\begin{proof}
Let $S$ be a group of size $k$. Then we have $N\geq |S| = k$ where $N$ is total number of players. Since each player can
contribute at most 1 exclusive volume we have $A_S\leq |S|\leq N$ where $A_S$ is volume covered by $S$. If $x\in S$ then we
have $c\leq E[u(x)] = \Big(\frac{M}{|S|}\Big)^{1-r}\cdot \frac{A_S}{A}\leq \Big(\frac{M}{|S|}\Big)^{1-r}\cdot \frac{N}{A}$.
Therefore we have $\Big(\frac{M}{|G|}\Big)^{1-r}\geq \frac{cA}{N}$ which implies $M\geq
|G|\Big(\frac{cA}{N}\Big)^{\frac{1}{1-r}} = k\Big(\frac{cA}{N}\Big)^{\frac{1}{1-r}}$.
\end{proof}

\end{document}